\documentclass[12pt]{article}
\usepackage{booktabs} 
\usepackage{graphicx}
\usepackage{amsthm}
\usepackage{amsmath} 
\newtheorem{theorem}{Theorem}
\usepackage[utf8]{inputenc}
\usepackage{epstopdf}
\usepackage[margin=0.75in]{geometry}
\usepackage{float}
\usepackage{natbib}
\usepackage{setspace}
\usepackage{subfig}
\usepackage{threeparttable}
\usepackage{url}
\usepackage{arydshln}
\usepackage{tikz}
\date{}
\usetikzlibrary{matrix,arrows,decorations.pathmorphing}

\begin{document}
\title{Long-run User Value Optimization in Recommender Systems through Content Creation Modeling}

\author{Akos Lada, Xiaoxuan Liu, Jens Rischbieth, Yi Wang, Yuwen Zhang \\ (Facebook)}

\vspace{0.5in}

\maketitle

\begin{abstract}
Content recommender systems are generally adept at maximizing immediate user satisfaction but to optimize for the \textit{long-run} user value, we need more statistically sophisticated solutions than off-the-shelf simple recommender algorithms. In this paper we lay out such a solution to optimize \textit{long-run} user value through discounted utility maximization and a machine learning method we have developed for estimating it. Our method estimates which content producers are most likely to create the highest long-run user value if their content is shown more to users who enjoy it in the present. We do this estimation with the help of an A/B test and heterogeneous effects machine learning model. We have used such models in Facebook's feed ranking system, and such a model can be used in other recommender systems.
\end{abstract}


%
%

\section{Introduction}

Content recommender systems have become a large source of how people consume content, may it be part of social media (eg Youtube, Facebook, Instagram, Tiktok) or other media types (eg news aggregators such as Google News). These systems usually take a high number of available pieces of content (items) that are eligible to be viewed by a person using the service when they log in, and they rank these items in a way that maximizes relevancy and value to the user. General methods to achieve this ranking include building prediction models to predict immediate interaction events such as likes, comments, or whether a post is worth the user's time, either through methods such as logistic regressions or neural networks (\citealt{hinton1986}), or taking estimating the best rank order of posts (\cite{liu2009}). However these prediction models maximize only the immediate value of content seen to the user. 

To maximize the \textit{lifetime} value the user derives from using the recommender system and sees a piece of content, we need to augment the recommender system to take dynamic effects into account. The novelty of this paper that it models such dynamic effects and explainsthe intertemporal nature of the optimization problem, and lays out a scaleable machine learning solution that we have successfully deployed at Facebook's feed ranking. 

More concretely, we model the scenario that a user's actions affect what content is going to be available to them in the future. If they engage with or otherwise give positive feedback on a piece of content, some content creators are more likely to create more of this type of content in the future, while some others are impervious or less likely to change their behavior (\citealt{gneezy2011}). In addition, users' actions can also create positive externalities (\citealt{pigou2017}) for each other, ie one user can create value for other users: if user $a$ likes a piece of content producer by producer $b$, which encourages producer $b$ to create more content, then if a third user $c$ is eligible to see all content that $b$ produces, user $c$ will also benefit from user $a$'s actions. In this paper we work out the solution to how to maximize user value in our recommender system through taking the content production encouragement and user-to-user spillover mechanisms into account, and thereby creating the maximal user value for each user in the ecosystem.

The solution we developed and have deployed to Facebook App's Feed ranking system is based on a two step-process. First, to find the optimum user satisfaction in our ecosystem we run an A/B test (\citealt{kohavi2013online}) that increases the distribution of randomly selected producers' posts. Then in a second step we use machine learning to find heterogeneous treatment effects in this experiment in terms of posting/engagement-received, ie we estimate a model about which content producers are most likely to increase their production and create further valuable content if they receive engagement. We can model this through heterogeneous effects machine learning such as a three-tree model, which we use at Meta and we lay out the details in the paper below. After the model is verified to work in a test set, we deploy the heterogeneous effects machine learning model in our ranking system, rewarding content producers that create good content in the future, in order to give an optimum amount of long-run user happiness to our users. The model, which scores all content producers, can be continuously retrained on a small producer holdout where producers do not receive treatments.

Further extensions can be used to take into account even the quality of the content produced by various producers through inventory demotion experiments. Users' other actions beyond posting quality can also be modeled.

Our method can be used in any recommender system or content ranking platform where content producers change their future behavior based on outcomes in the present, and where we want to maximize the long-run value the system delivers to its users. This can range widely from social media, e-commerce and search systems to even systems selling services such as ride-sharing or holiday rentals. We derive conditions under which our method can have outsized importance for any such platform: (1) the more forward-looking the decision-makers are and thus the more they value the future highly; (2) if users derive the most short-run value from different content than the most long-run value; (3) the stronger content producers change their behavior depending on distribution they receive from the recommender system; (4) the more users' choices create positive externalities for each toher. If any one or multiple of these conditions is strongly true in a recommender system, it is worth considering the algorithm described in this paper.

\section{Context}

\subsection{Illustrative Model}

The dynamics of our system can be illustrated with a simple `micro-economic model' (\citealt{mascolell1995}). There are two people using the recommender system, call them viewers: $v_1$ and $v_2$, who consume content from two producers, $p_1$ and $p_2$ in two periods, $t$=1 and $t$=2. Assume both viewers have limited time to spend on the recommender platform, so they are each able to view a single piece of content. Also assume that viewers' satisfaction can be expressed as $l_1(t)$ and $l_2(t)$ for viewer 1 and viewer 2 and is captured simply as how much these two viewers like what they see in period $t$. For simplicity $l_1(t),l_t(2)\in[0,1] $, where $l_1(t)=0$ means viewer 1 gets no satisfaction from the content they see in period $t$ and $l_1(t)=1$ means they get the maximum possible satisfaction, but $l_1(t)$ can also take on intermediate values. We also assume that consumption is more valuable in the first period than the second period, ie all satisfaction in the second period are discounted by a $0<\beta<1$  discount factor. This means the satisfaction of viewer $i\in\{1,2\}$ can be expressed as:
$$S_i=l_i(1)+\beta*l_i(2),$$
and the recommender system's objective function is to maximize $\sum_{i\in\{1,2\}} S_i.$ We also assume that when a user views a piece of content they like, ie $l_i(t)>0$ they engage with the piece of content in a way that is revealed to the content producers, eg they like the content, they comment on it, or otherwise reveal their satisfaction.

Note that the simple non-inter-temporal maximization would be simply to rank in each period for each viewer the content on top which is most likely to be liked the most by the viewer. However with dynamic incentives, the optimal solution can change. Assume that content producer 1 creates no content in the second period regardless of engagement received in the first period, whereas producer 2 is more interested in engagement received in the first period: they produce no content in the second period if they do not receive any engagement in the first period and produce a single piece of content if (s)he receives engagement from at least 1 user in the first period. Now if the value of the first producers' content is higher for both our viewers than the second producer's piece of content $v_1>v_2$ then the short-horizon recommender system would show simply producer 1's content to both of our viewers. However as long as $(1+\beta)*v_2>v_1+0$ for at least one of our viewers, the inter-temporal optimizer would instead show producer 2's content to at least one of the viewers, and creating positive externalities for the other user (new content inventory in the second period). In words, as long as the loss is bigger from not seeing any content in the second period than the short-term cost of not seeing the optimal content in the first period, it is worth taking content production incentives into account in our recommender system.

By using A/B testing and machine learning together we can learn that producer 2's second period decisions depend highly on their first-period distribution. More generally, we can learn which producers are most likely to create more content when they receive engagement. In an A/B test where we give more distribution to randomly selected producers we would learn that producers such as producer 2 increase their production in treatment compared to control when they receive engagement, whereas producers such as producer 1 do not increase their production. We can learn this machine learning model to target the producers in question.

\subsection{General Model}

To generalize our model, assume that we have $N$ users in our recommender system who are expected to visit the platform over $T$ periods ($T$ can be infinity), each time consuming $J$ pieces of content, so our objective is to maximize the sum of discounted long-run user value in our system:
$$\sum_{i=1}^{N}\sum_{t=1}^{T}\sum_{j=1}^{J}\beta^t l_{itj}(R_{it}),$$
where $l_{itj}\in[0,1]$ represents the value user $i$ derives in period $t$ from the content they see in $j$'th position in the recommender system. This value depends on the ranking $R_{it}$ our system gives to the user $i$ in period $t$. Note that for simplicity we assume a user gets the same value from the same post if they see it in period $t$, but assume that there are a large amounts of available posts for each user $i$ in each period $t$, so a piece of content need to be ranked within the top $J$ posts for the user to see it and get satisfaction from it.

In this general model the logic laid out above means that the optimal solution to maximize $\sum_{i=1}^{N}\sum_{t=1}^{T}\sum_{j=1}^{J}\beta^t l_{itj}(R_{it})$ is not simply to maximize the short-term value for each user. We state this in the following theorem: 

\begin{theorem} The $R'_{it}$ that maximizes $short-term$ and only $individual$ user value, ie the $R_{it}$ that solves $arg max_{Rit}\{\sum_{j=1}^{J}\beta^t l_{itj}(R_{it})\}$ for each $i$ and $t$ might \textit{not} be the global optimum $R''_{it}$:
$$arg max_{R''it}(\sum_{i=1}^{N}\sum_{t=1}^{T}\sum_{j=1}^{J}\beta^t l_{itj}(R''_{it}))$$. 

In fact, $R'_{it}\neq R''_{it}$, ie the short-run optimum is not the long-run optimum for at least some $i$ and $t$, when content availability in the future depends on prior periods' consumption behavior, ie $l_{itj}(R_{it})$ influences $l_{itj}(R_{i't'})$ for some $t'> t$ and some $i'$ ($i'$ can be the same as $i$). So there is `spillover' from one period to another. If this holds, the short-run optimum will be different than the long-run optimum (ie $R'_{it}\neq R''_{it}$) as long as we can find some ranking order $R''_{it}$ for some $i$ and $t$ that compared to the short-run maximizing option $R'_{it}$ satisfies the following inequality:
$$\sum_{i=1}^{N}\sum_{j=1}^{J} l_{itj}(R'_{it})-\sum_{i=1}^{N}\sum_{j=1}^{J} l_{itj}(R''_{it})<$$  $$\sum_{i=1}^{N}\sum_{k=t+1}^{T}\sum_{j=1}^{J}\beta^{k-t} l_{ikj}(R''_{it})-\sum_{i=1}^{N}\sum_{k=t+1}^{T}\sum_{j=1}^{J}\beta^{k-t} l_{ikj}(R'_{it}),$$
where the left hand side reflects the short-run benefit in period t from ranking with the short-run optimum $R'_{it}$ compared to ranking $R''_{it}$, and the right hand side represents the long-run benefits from ranking  ranking $R''_{it}$ compared to the short-run ranking $R'_{ik}$ from $t+1$ period to $T$, ie the future consequences of ranking by different policies in period $t$. Essentially this condition means that there is at least one ranking, where the discounted gains in the long-run outweigh the costs of a sub-optimal short-run ranking.
\end{theorem}
\begin{proof}
Following the logic from the simple two period model above, we can see that the assumption laid out in the theorem is a necessary assumption for the long-run global maximum might not be the same as the short-run individual-level 1-period-horizon maximum. The assumption states that there is spillover from one period to another (ie one period affects another period's available ranking). This requirement essentially means that $l_{itj}(\cdot)$ is not just a function of $R_{it}$ but also of possibly many $R_{i't'}$'s where $t'\in [1,t-1]$ and $i'$ can be any user including $i$. Thus $l_{itj}(\cdot)$ for any period $t>1$ is more completely written as $l_{itj}(R_{it},R_{i't'},...)$ where every single user $i'\in [1,N]$ and every single $t'\in[1,t-1]$ is an argument of the $l_{itj}(\cdot)$ function. This clearly means that since any $R_{it}$ can affect future $l_{itj}(\cdot)$, and there is at least one that indeed does affect future $l_{itj}(\cdot)$'s, the one-period optimization could be different from the multi-period optimization, as the former ignores the effect on multiple $l_{itj}(\cdot)$'s.

Next, as long as the assumption stated in the theorem holds, we know that intertemporal optimization $might$ be different but this assumption is not enough. For instance it could be that the discount factor is so low that the short-run sub-optimal ranking decisions for long-run gains are never worth it. Or it could be that the producers who create the most long-run valuable content also create the most short-run valuable content, so the short-run vs long-run discounted optimization yield the same result.

On the other hand, if the condition stated in the theorem holds, then there must be at least one ranking that we'd choose differently if we take the long run into account. To see this, assume the opposite: ie even though the inequality in the theorem holds both the short run and the long run optimizer chooses the same policies. But in this case given that the inequality holds, it means that if $R'_{it}$ is the optimal policy then we can find a policy $R''_{it}$ which yields  positive improvements for the long-run optimizer because $\sum_{i=1}^{N}\sum_{k=t+1}^{T}\sum_{j=1}^{J}\beta^{k-t} l_{ikj}(R''_{it})-\sum_{i=1}^{N}\sum_{k=t+1}^{T}\sum_{j=1}^{J}\beta^{k-t} l_{ikj}(R'_{it})-(\sum_{i=1}^{N}\sum_{j=1}^{J} l_{itj}(R'_{it})-\sum_{i=1}^{N}\sum_{j=1}^{J} l_{itj}(R''_{it}))$ is positive. So $R'_{it}$ cannot be the optimal long-run policy. If instead $R''_{it}$ is the optimal short and long-run policy then we can find an $R'_{it}$ such that $\sum_{i=1}^{N}\sum_{j=1}^{J} l_{itj}(R'_{it})-\sum_{i=1}^{N}\sum_{j=1}^{J} l_{itj}(R''_{it})$ is positive, hence $R''_{it}$ cannot be the optimal short-run policy. We have arrived at a contradiction, hence the optimal short-run vs long-run policy cannot be the same as long as the condition in the theorem holds. QED.
\end{proof}

\subsection{Comparative Analysis: when is long-run ecosystem modeling relevant?}

Looking at the important condition above that needs to hold for short-run optimization and long-run optimization to differ, we can analyze each parameter to understand under what circumstances our model is more important to apply in the content recommender system:
$$\sum_{n=1}^{N}\sum_{j=1}^{J} l_{itj}(R'_{it})-\sum_{n=1}^{N}\sum_{j=1}^{J} l_{itj}(R''_{it})<$$ $$\sum_{n=1}^{N}\sum_{k=t+1}^{T}\sum_{j=1}^{J}\beta^{k-t} l_{ikj}(R''_{it})-\sum_{n=1}^{N}\sum_{k=t+1}^{T}\sum_{j=1}^{J}\beta^{k-t} l_{ikj}(R'_{it}).$$

First note that the higher $beta$ or $T$ is the more important it is to model the long-run and not just the short run. Ie the more forward-looking our decision-makers are, the more likely long-run modeling is needed.

Second, the bigger the difference between $l_{itj}(R'_{it})$ and $l_{itj}(R''_{it})$, the more likely it is that we need to model the short run and long run. In other words, if users get very similar value out of different rankings, the simpler short-run optimization suffices, but if some type of content yields especially big values to users, it is worth modeling the long-run more.

Third, the more $l_{it'j}(R_{it})$ is responsive to $R_{it}$ where $t'>t$, ie the bigger $d(l_{it'j}(R_{it}))/d(R_{it})$ is in absolute value, the more important it is to model the long run, especially if some producers respond very differently than others. In other words if future production is very dependent on the present engagement, it's worth modeling the long run in detail.

Finally, if there are a lot of interactions in our network so $$d(l_{i't'j}(R_{it})) / d(R_{it})$$ is high in absolute value, where $i\neq i'$ and $t'>t$, the more long-run ecosystem value is important. Our users still get their individually maximum long-run user value but positive spillovers mean they create value to each other too. To illustrate this case, imagine that a user's actions lead to more production of a certain type of inventory in the future, and that inventory is highly beneficial to other users. With this user-to-user spillover channel being strong, it's again highly worth modeling long-run network effects.
 
\section{Model}

To model the heterogeneous effects and find the producers for whom engagement received is most meaningful, we first run an A/B test because we need experimental variation for our causal inference (\citealt{rubin2011causal}, \citealt{pearl2009causality}). A/B tests are standard in the technology industry nowadays. In the A/B test a randomly selected set of producers are given more distribution, while a control set of producers are receiving no change in their content's ranking. Such A/B tests might have already been run in the past (\citealt{peysakhovich2018learning}), in which case we can use the already existing experiment. Once we have the experimental data collected over a reasonable period (eg a week), then we use a three-tree method to estimate the heterogeneous treatment effects (similar methods are laid out in \citealt{kunzel2019metalearners} and \citealt{wager2018estimation}). All data collected has been aggregated and de-identified.

First we randomly select some producers to be our evaluation group and do not use this test group until our models below are fully built. Then we take a large set of features describing content producers such as how many followers they have, how often they produce content, etc, all de-identified. It is important that each of these features need to be observed from before the experiment started. Then we take the features to predict individual level outcomes for each producer in parallel on the treatment group and the control group, and fit a content-producer-level model, $m_{treatment}$ and $m_{control}$. This modeling can be done with any off-the-shelf machine learning model, and we found that gradient boosted decision trees (GBDT's, \citealt{friedman2001greedy}) or neural networks (\citealt{hinton1986}) are good choices to capture the non-linearities with high performance. Note that our two models allow us to estimate the individual treatment effect for any single producer $i$, by evaluating the difference $$m_{treatment}(i)-m_{control}(i)$$. 

To go beyond individual treatment effects and generalize $$m_{treatment}(i)-m_{control}(i)$$ to the entire population, we take $$m_{treatment}(i)-m_{control}(i)$$ for all $i$'s in either treatment and control (only in our training data, we do not make any use of the test data yet). Then we fit a gradient-boosted decision tree or a neural network model on $$m_{treatment}(i)-m_{control}(i)$$ using the full set of producer features that we have built our \\ $m_{treatment}(i)$ and $m_{control}(i)$ with. Call this model $m_{difference}$. This $m_{difference}$ will in fact tell us what the predicted treatment effect (or sensitivity) is for any producer group. To interpret this model we can use standard methodologies such as feature importance plots.

Finally, once the model is fit, we can use the evaluation set to evaluate our model. Using $m_{difference}$ we can divide producers into $high$ and $low$ group based on a percentile cut-off (eg p20, p50, p80) of the model. Then we compare the $high$ group's production with the $low$ group's to validate if the $high$ group is indeed statistically significant to produce more content as a result of engagement received.

\section{Example and Results}

We have used the above method at Meta's Facebook News Feed. At News Feed, professional creators can accumulate followers through good content creation, which followers can interact with in various ways: liking, commenting, clicking, watching. Content creators' content can also be recommended to non-followers, however this is governed by a separate ranking system so we limited our work here to `connected distribution', ie distribution to followers.

We have selected 2$\%$ of producers at random and increased their distribution to their followers. Comparing with another randomly selected 2$\%$ producers without the treatment, the experiment resulted on average in +26$\%$ increase in the likes these producers received and +9$\%$ increase in the comments these producers received compared to their control counterparts. As a result, in the experiment we see that on average the treatment producers have increased their posting by +0.19$\%$. 

To find the heterogeneous effects, we then fit our three-tree model. All three models finished successfully and parameters in the models were set to reach the lowest point of the learning curve. They are all gradient boosted decision trees given the outcome variable being absolute outcomes and to capture non-linear relationships. The most important features in the first and second model are around the audience size of the producers and their general activity levels. The third model, the most key model, is trained with 242,083 examples. We use 1000 rounds of training (found through parameter tuning), 128 leaves and a feature sampling rate of 0.75. The model loss is 10,882.6 which is 73.58\% of the baseline loss of 14,788.9, showing that the model is learning in a valid way and the learning curve flattens out at this normalized loss rate. Cross-checking the top features in the model, and putting users into two groups 'low' and 'high' based on this feature, we see statistically significant difference between the two groups (0.11\%+-0.17 vs 1.5\%+-1.1), further underpinning the validity of the model's learnings. For further verification a follow-up pair of experiments can be run, where we increase the distribution of equal numbers of $high$ and $low$ producers and verify that the production gains in the $high$ group are higher.

\section{Deployment}

Once we built our heterogeneous treatment effects machine learning model, we can deploy our model for the entire ranking system. Note that this means assigning each producer in the recommender system a score. The score does not necessarily need to be real-time since the model should be quite stable because it is unlikely that a producer is found to be changing their behavior a lot for engagement received one day but much less so the next day. In fact such instability should be avoided as it is likely important for producers to be able to roughly predict the distribution they get, otherwise they will not invest the effort into making high-quality content, not knowing whether they'd get any meaningful amount of distribution. To increase stability, we can train our machine learning model above with only stable features such as how many followers a producer has.

\begin{figure}[h]
\caption{Deployment system example. The model is learned through a production machine learning system such as Meta's FBLearner system. The model can be learnt on a blanket boosting experiment before deployment and on a holdout after deployment. The model outputs producer scores which then feed into production ranking. Viewers are served by the content ranking system which ranks producers' content through the producer score. The producer score is created by applying the FBLearner model on each individual producer's features.}
\centering
\includegraphics[width=0.5\textwidth]{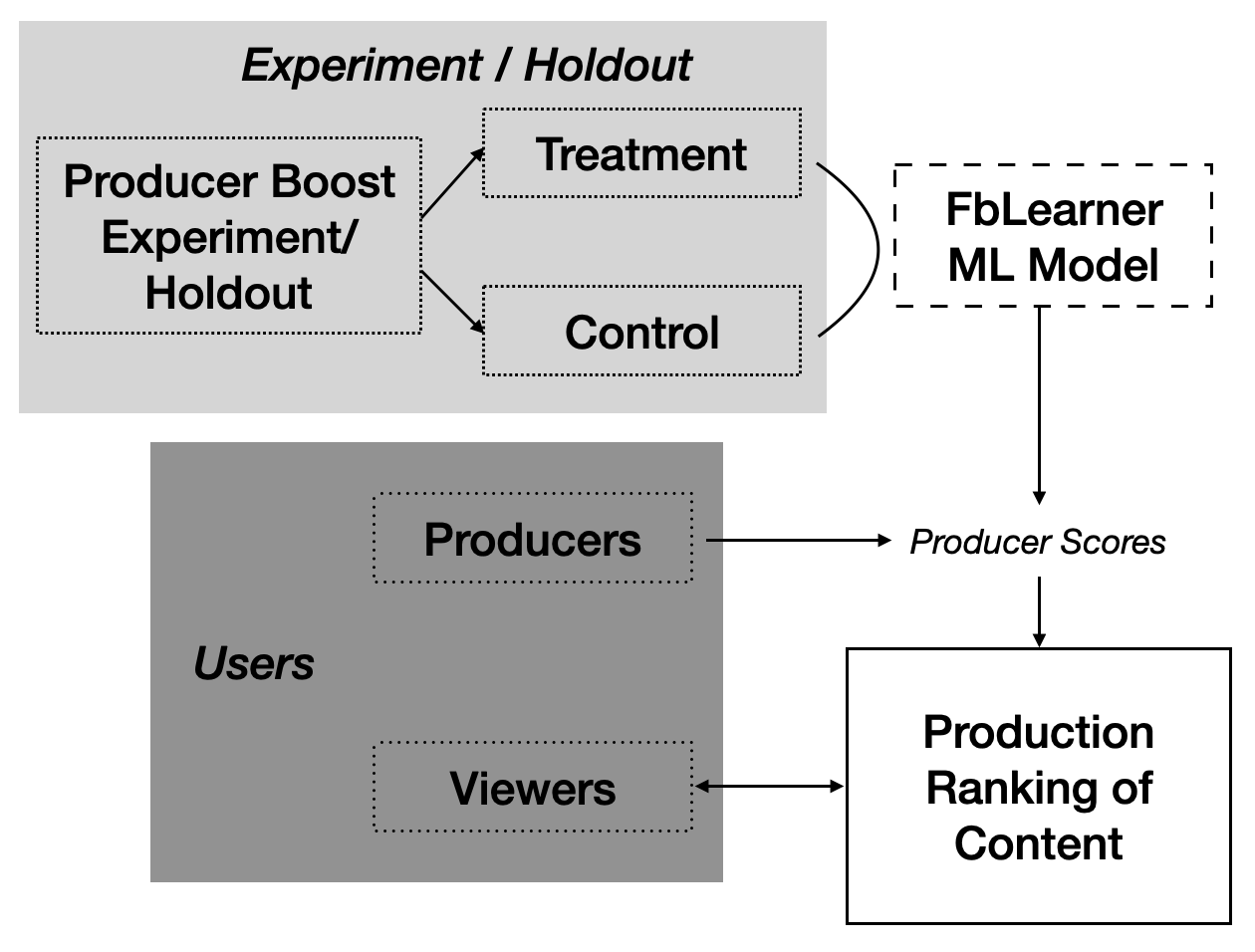}
\end{figure}

Despite the stability, it is also important for our model to get retrained. As the model is trained on an experiment, it is important to have a small holdout at launch-time, where some producers will not be receive the heterogeneous score treatment when they create content. To be more precise, we should give these `holdout producers' the average score, so that they do not get too high or too low distribution once the model is deployed for the rest of the producers. The model can then be retrained periodically from data on the holdout. In deployed models we found a one-week retraining to be reasonably well performing. The holdout will also tell us when the model has become outdated and does not increase user value any more, in which case we can deprecate the model.

An alternative deployment strategy is to use the model to build into a goal metric that counts the engagement given to producers weighted by their heterogeneous effects score, but not directly use the model in the day-to-day recommender system. If this route is taken, we can evaluate all changes that are proposed to be made on the recommender system based on this goal metric. The advantage of this method is that if the model retrains, there is no immediate, possibly discontinuous, change in production ranking. Instead the method just updates the evaluation metrics, and any new model that is under consideration to be launched into the recommender system will be evaluated accordingly. We have used variants of this approach too at Meta with reasonable success.

\bibliographystyle{ACM-Reference-Format}
\bibliography{creators}

\end{document}